\documentclass{article}
\usepackage[utf8]{inputenc}
\usepackage{amssymb}
\usepackage{amsmath}
\usepackage{amsthm}
\usepackage{amsfonts}
\usepackage{eucal}
\usepackage{algorithm,algorithmic}
\usepackage[letterpaper, margin=1in]{geometry}
\usepackage{graphicx}
\graphicspath{ {./images/} }
\usepackage{mathtools}
\usepackage{natbib}
\DeclarePairedDelimiter\ceil{\lceil}{\rceil}

\newtheorem{theorem}{Theorem}

\newtheorem{lemma}{Lemma}

\title{Colonel Blotto Game: An Analysis and Extension to Networks\footnote{I thank Kevin Ren for introducing me to the Colonel Blotto game, and Prof. Daniel Kane for reading and offering valuable suggestions on an earlier draft. This project was originally conceived and executed while I was in eighth grade for submission to the 2021 Greater San Diego Science and Engineering fair.}}
\author{Sidarth Erat\footnote{La Jolla High School, La Jolla, email: sidarth.erat@gmail.com}}
\date{\today}
\date{}
\usepackage{setspace}
\begin{document}

\maketitle
\begin{abstract}
    The Colonel Blotto game, introduced by Borel in the 1920s, is often used for modeling various real-life settings, such as elections, lobbying, etc. The game is based on the allocation of limited resources by players to a set of fields. Each field is ``won'' and a corresponding field-specific value is obtained by the player who sends the most resources. 
    In this paper, we formulate a discrete Blotto game played on a general \textit{accessibility network} (i.e., the bipartite graph made of players and the fields they can allocate resources to). The primary goal is to find how the topology of the accessibility network controls the existence and uniqueness of equilibrium allocations, and how it affects the fraction of fields that are entered and the average payoff of players at equilibrium. 
    
    We establish that, in a 2-regular topology, when the values of fields are close enough and the number of players is not a multiple of 4, then there is a unique equilbrium. We also prove that players are better off and fields are more likely to be entered in a regular topology than a random topology. We find numerically that dispersion of field weights negatively affects average player payoff. The main contribution is a framework for analyzing contests where players are permitted access to some (but not necessarily all) venues of competition.
    
\end{abstract}

\large

\section{Introduction}
Multiplayer resource allocation contests are ubiquitous in daily life. 
The Colonel Blotto game, introduced by Borel almost one hundred years ago, is one of the most basic models of resource allocation contests. 



As an example, consider a population of cheetahs roaming over a large domain. In the classical Blotto model, cheetahs allocate their energy to hunting gazelles in various locations, and the cheetah that allocates the most energy in a given location gets the gazelle in that location. However, the classical model discounts the possibility that not every location is accessible to every cheetah. Specifically, a cheetah might be able to hunt only in their ``neighborhood'' hunting grounds.

There is a natural way to correspond these restrictions on players' movement to an ``accessibility network,'' which is a bipartite graph with the set of players as one part and the fields as the another part. This paper aims to find how these network restrictions on which venues can be accessed by each player affect the equilibria of the game and ``field survival rate'' (in the phrasing of our example, the number of locations where no cheetahs chooses to hunt). 

In order to focus our attention to the network structure, we consider a discrete Boolean Blotto game (i.e., each player has a single resource that cannot be split between fields). We first focus on 2-regular networks and give conditions for the existence of unique Nash equilibria. In this case, it is also found that at equilibrium, greater heterogeneity in field weights leads to lower average payoffs and higher survival rates. We then direct our attention to comparing and contrasting k-regular networks with randomly generated networks, with the additional assumption that fields have the same weight. In this case we find that in a random topology, survival rate is higher and average payoffs are lower.

The rest of the article is organized as follows: In \S\ref{lit-review} we offer a review of the literature. We introduce and formalize the model in \S\ref{model}. \S\ref{2-regular-case} solves the case of a 2-regular network, and \S\ref{random-vs-regular} offers the comparison between random and regular network topologies. In \S\ref{discussion-and-conclusions}, we discuss the results.

\section{\label{lit-review} Literature Review}

The Colonel Blotto game has been extensively researched for a century, and the literature is vast. So, we offer a selective review, focusing on a succession of key studies that illustrate the gradual development and refinement of the Colonel Blotto model.

The Colonel Blotto game was originally introduced by Emile Borel in the early 20th century. The game he analyzed involved 2 players, each with a fixed amount of resources. The players allocate these resources to battlefields of different weights, each won by the player that sends the most resources to it. Each player tries to maximize the total weight of all fields won. The game, which, on the surface, appears quite simple, has still not been fully solved for arbitrary weights and resource levels. Below are several variants of the game analyzed over the last 100 years:

Before the development of formal modern game theory, in 1921, \cite{borel1953theory} established that the discrete 2-player Colonel Blotto (where both players have $B$ discrete resources) has a pure strategy solution for all $B < 7$. In the 1920's, Borel and Ville further analyzed the game, presenting mixed strategy equilibria (i.e., where strategies are probability distributions over the set of possible plays) for the two-player, equal weight, equal-resource case of the continuous Colonel Blotto game. Borel is responsible for the ``hex solution,'' an equilibrium based on the random picking of a strategy $(x_1, x_2, x_3)$ from a recursively-defined hexagonal lattice.

In a memorandum of \cite{grossman1950}, the 2 field case and the 3 field, equal resource case was fully solved. This was the first step in finding equilibria of two player Blotto games.

\cite{roberson2006colonel} fully solves the equal weight case of two person Blotto, assuming each person can access every venue of competition. Its major methodological contribution is the use of copulas to couple marginal equilibrium distributions, which can be done by applying Sklar's theorem.



\cite{boixadserà2021multiplayer} analyzes multiplayer Blotto games in general. The paper focuses on a Boolean Blotto game, somewhat similar in structure to the variant used in this paper. The resources sent to any given field in Boolean Blotto are always in the set $\{0,1\}$. This is the first significant paper that fully solves a Boolean variant of multiplayer Blotto. However, the paper uses a symmetric variant, rather than one with variable field valuations and/or restrictions on troop movement, the latter of which is the primary focus of the current study.



To summarize, in Borel and Ville's original analysis of the game, solutions were introduced for very specific cases, including the equal weight, equal resource case. \cite{grossman1950} introduced full solutions for the 2 field, 2 player case, along with the 2 player, 3 field, equal resource case. Later, \cite{roberson2006colonel} utilized copulas to combine marginal equilibrium distributions for the 2 player, equal-weight case into a viable multivariate distribution. Recent research has found fairly elementary equilibria for certain subcases of the game, including the multiplayer variant (both continuous and Boolean). A solution for equilibria of the Boolean variant is given in \cite{boixadserà2021multiplayer}. Through extensive analysis over the decades, significant progress has been made both in the full solution of the 2-player case of Colonel Blotto's game and the solution of various discrete and/or multiplayer Blotto variants. The current study adds to this literature by considering an orthogonal dimension, namely possible constraints on the players ability to access different fields, and takes a first cut at understanding how the ``accessibility network'' plays a role in determining the equilibrium outcomes.

\section{\label{model}The model}

 Consider a Blotto game on $n$ players (which are henceforth called ``cheetahs'') and $m$ fields (which are henceforth called ``gazelles''). Each cheetah has a finite "hunting budget" (the total amount of effort it can exert hunting).
Let these cheetahs have hunting budgets $b_1, b_2, \cdots b_n$, and let the gazelles have weights $w_1, w_2, \cdots w_m$.

Define the \textit{accessibility matrix} to be an $n \times m$ matrix with entries $a_{ij}$ such that $a_{ij} \in \{0,1\}$, and $a_{ij}=1$ iff the $i$th cheetah is connected to the $j$th gazelle (i.e., $j$th gazelle is accessible to cheetah $i$).

Furthermore, let the cheetah $i$'s budget allocation be represented by $\vec{x_i} \in \mathbb{R}^m$ with entries $x_{i1}, x_{i2}, \dots x_{im}$ such that:
\begin{enumerate}
    \item[(1)] $\sum_{j=1}^m x_{ij} \leq b_i$
    \item[(2)] $0 \leq x_{ij} \leq b_i a_{ij}$
\end{enumerate}
with the interpretation that $x_{ij}$ is the budget allocated by cheetah $i$ to gazelle $j$.

We assume that the cheetah that exerts the most effort hunting a specific gazelle wins that gazelle. That is, gazelle $j$ is won by cheetah $i$ if $x_{ij}>0$ and for all $p \neq i$, $x_{ij} > x_{pj}$. If there is a tie for highest effort among $k$ cheetahs, each cheetah gets value $w_jr^k$, for a fixed constant $r$ (where $0\le r <1$).

In other words, the $i$th cheetah has total winnings:
$$\pi(i) = \sum_{j=1}^m w_j \prod_{p \neq i} g(i, j, p)$$ 
where $$g(i, j, p) = \begin{cases}
    0  \text { if  } x_{ij} < x_{pj} \\
     1 \text{ if } x_{ij} > x_{pj} \\
       r \text { if  }x_{ij} = x_{pj}, \text{ for some fixed }  0 < r < 1 \\
       \end{cases}$$
The cheetahs play mixed strategies over the set of all possible $\bar{x_i}$. In other words, let $$X_i = \{ \vec{x_i} | 0\leq x_{ij} \leq b_i a_{ij} ; \sum_{j=1}^m x_{ij} \leq b_i \} $$ Hence, a strategy for cheetah $i$ is a distribution over the set $X_i$.

We shall consider only the case where $b_i =1$ for all $i$, and additionally assume that $x_{ij} \in \{0,1\}$. More specifically, this paper focuses only on a discrete Boolean Colonel Blotto variant (i.e., each player has only one unit of resources and splitting of resources between fields is disallowed).

\section{\label{2-regular-case}Case of 2-regular topology}

Firstly consider the case of a 2-regular topology, specifically a network topology where $m=n$ and, for each $i$, cheetah $i$ can hunt only gazelles $i$ and $i+1$ (indices are taken mod $n$). In this case, every mixed strategy of cheetah $i$ can be represented as some number $p_i$ such that $0 < p_i < 1$, so that cheetah $i$ goes to gazelle $i$ with probability $p_i$ and goes to gazelles $i+1$ with probability $1-p_i$.
For any mixed strategy equilibrium, the expected payoff from playing either pure strategy must be the same, or cheetah $i$ would lower its expected payoff by mixing. 

We first characterize how interior point mixed strategy equilibria should look if they exist.

\begin{lemma}
If the tuple of strategies $(p_1, p_2, \dots p_n)$ is an interior point equilibrium  (i.e., equilibrium where $0<p_i<1$), then for all $i$,  $$ w_ip_{i-1} + w_{i+1}p_{i+1} = 2w_{i+1}-w_i$$ 
Note: Non-interior point equilibria need not satisfy the set of equations.
\end{lemma}

\begin{proof} Note that if some tuple $(p_1, p_2, \dots p_n)$ \footnote{indices are taken mod $n$} is an equilibrium, then for every $i$, the expected payoff of exerting all its effort hunting gazelle $i$ is:
$$ w_i \left( 1 * p_{i-1} + \frac{1-p_i}{2} \right) = \frac{w_i}{2} (1+ p_i)$$
The expected payoff of spending its entire budget on gazelle $i+1$ is:
$$ w_{i+1} \left( 1- \frac{p_{i+1}}{2} + 1 * (1- p_{i+1}) \right) $$
 Equating these two expressions, for all $i$,
$$ w_ip_{i-1} + w_{i+1}p_{i+1} = 2w_{i+1}-w_i$$ \end{proof}


The above lemma says nothing about the existence of interior point equilibria. So, we now analyze conditions for the existence and uniqueness of equilibria. To do this, we first prove the following auxiliary lemma.

\begin{lemma}
    Let $M$ be an $n\times n$ matrix such that $a_{ij}=1$ if $i=j-1$ or $i=j+1$, and $a_{ij}=0$ otherwise. Then, $M$ is invertible if and only if $n$ does not divide 4
\end{lemma}
\begin{proof}
If $4|n$, color column $\vec{x_i}$ red if $i \equiv 1, 2 \pmod 4$, and blue if $i \equiv 0, 3 \pmod 4 $.

In every row, with adequate rotation of columns to preserve color, the nonzero entries form a block: $[1 \text { } 0 \text { } 1]$. Clearly, however this is placed, each 1 must be colored differently. Hence, the sum of the blue colored entries is equal to the sum of the red colored entries. Therefore, the columns are linearly dependent and A is not invertible. 

Now consider the case where $n \not \equiv 0 \pmod 4$. Now let $t_1, t_2, \dots$ be some set of real numbers such that $\sum_{i=0} t_i\vec{x_i} = \vec{0}$. Note that by Lemma 1, for all $i$:
$$t_{i-1} + t_{i+1} = 0, t_{i+1} + t_{i+3} = 0 \implies t_{i-1} = t_{i+3}$$ 
where all values $i$ are taken $\pmod n$

For convenience, split this case into 3 subcases: (1) $n = 4a+1$ for some natural $a$, (2) $n = 4a+3$ for some natural $a$, and (3) $n= 4a+2$ for some natural $a$.

If $n =4a+1$, note that from basic modular arithmetic, for all $i$, $$t_i = t_{i+4} = \dots = t_{n+i-1} = t_{i-1} = t_{i+3}$$
Plugging in $i+3$, $t_{i+3} = t_i = t_{i+6}$.
Hence, $t_i = t_{i+6} = t_{i+2} = -t_i$, so $t_i = 0$ for all $i$. Therefore, the columns of $M$ are linearly independent and $M$ is invertible.

If $n = 4a+3$, note again that: $$t_i = t_{i+4} = \dots = t_{n+i-3} = t_{i+1}$$
Plugging in $i+1$, $t_{i+1} = t_{i+2}$
Hence, $t_i = t_{i+1} = t_{i+2} = -t_i$. Therefore, $t_i= 0$ for all $i$. Therefore, the columns of $M$ are linearly independent and $M$ is invertible. 

Finally, if $n = 4a+2$, note that: $$ t_i = t_{i+4} = \dots = t_{n+i -2} = t_{i-2} = t_{i+2}$$
Hence, $t_i = t_{i+2} = -t_i$, so $t_i = 0$ for all $i$. Then, the columns of $M$ must be linearly independent, and $M$ is again invertible.

Therefore, for all cases where $4\not |n,$ $M$ is invertible.
\end{proof}
\begin{theorem}
When $n$ does not divide 4, there exists a unique interior point equilibrium for sufficiently small $\dfrac{(\max_j {w_j} - \min_j {w_j})}{\max_j w_j}$.
\end{theorem}

\begin{proof}

Without loss of generality, assume that $\max_j w_j = 1.$
We first claim that an equilibrium exists if $w_i=w_j$ for all $i,j$. Moreover, this equilibrium is $p_i = 1/2$ for all $i$. 

This is true since if every cheetah $i$, for all $i \ne 1$ goes to gazelle $i$ with probablity $1/2$ and gazelle $i+1$ with probability $1/2$, then for all strategies $p_1$ that cheetah $1$ has, the expected winnings $W (p_1)$ by playing $p_1$ is:
$$ (1/2) p_i + 1/2 (1 - p_i) + 1/4 (p_i) + 1/4 (1- p_i) = 3/4 $$
regardless of $p_i$. Hence, this strategy is indeed optimal and the equilibrium $(1/2,1/2,\cdots)$ is the solution of the set of equations in Lemma 1.


Moreover, the determinant of a matrix is a continuous polynomial in the entries, so the set of invertible matrices is open. Therefore, since the matrix $M$ (defined as in Lemma 2) is invertible when weights are equal and $4 \not | n$, the new coefficient matrix should still be invertible when $\max_j w_j - \min_j w_j$ is small enough. Therefore, the solution we found to this system must be the unique set of strategies that satisfies this constraint. Hence, this is the only equilibrium. 
\end{proof}

This result immediately gives an algorithm to obtain the Nash equilibrium, simply by solving the system given in Lemma 1. To better analyze the effect of dispersion in weights on each players strategy, let the weights of the gazelles be in arithmetic sequence with common difference $\epsilon$, and let gazelle $\ceil{n/2}$ WLOG have weight 1. By solving our model for various cases of $\epsilon$, we find that the average average cheetah payoff decreases as $\epsilon$ increases, as shown below when $n=11$ (left). However, the expected payoffs for specific cheetahs might increase or decrease, as shown when $n=11$:

\includegraphics[scale=0.2]{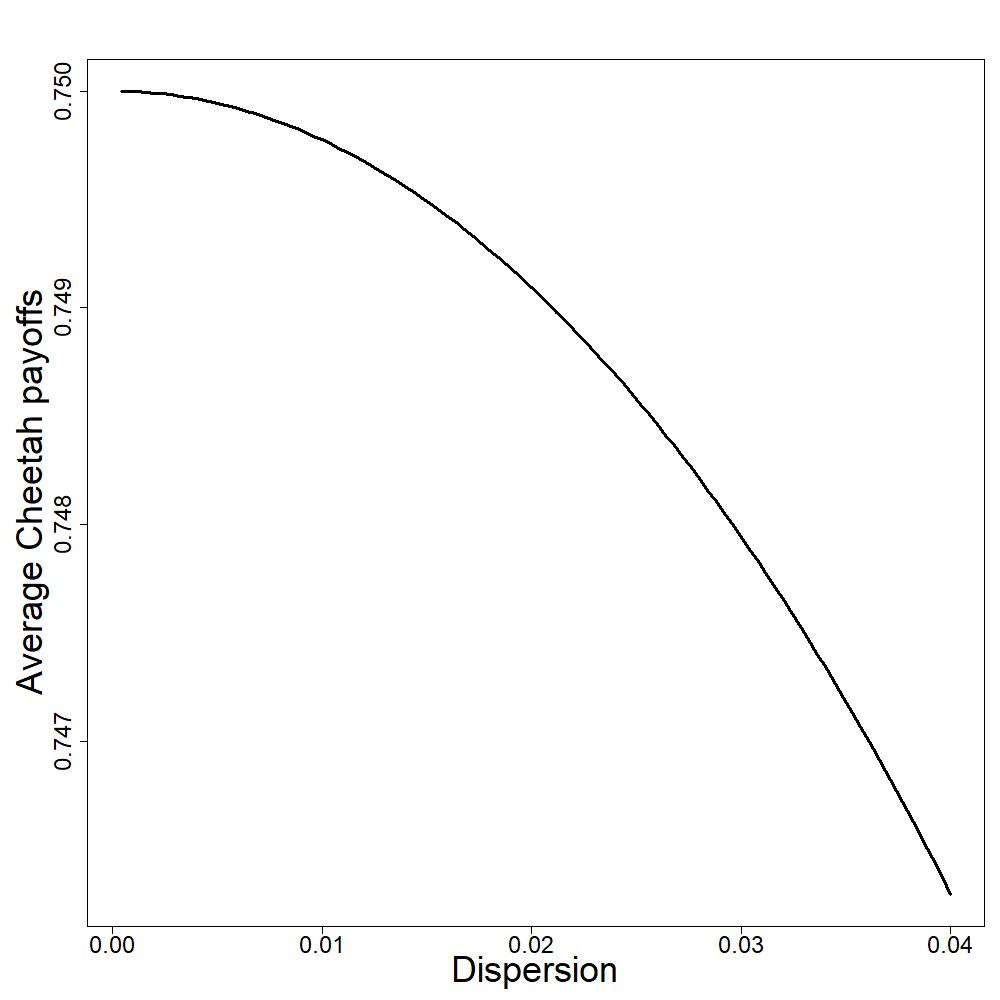} \includegraphics[scale=0.2]{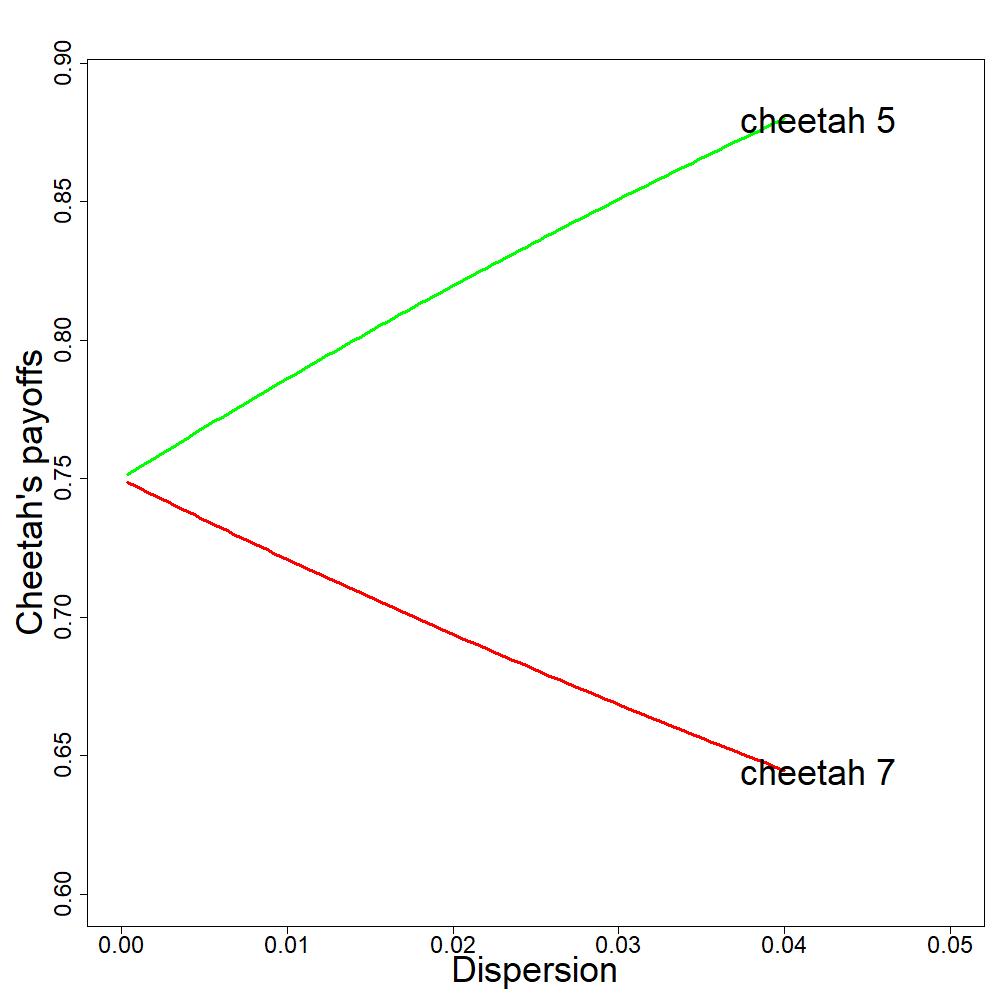}

\section{\label{random-vs-regular}Random vs. Regular topology}

We analyze two network topologies in this section, the $k-$ regular topology and $k-$ random topology, defined as follows: (i) In the $k-$regular topology, each cheetah $i$ is connected to $k$ gazelles, with these $k$ gazelles being $i,i+1,\cdots,i+k-1$ (mod $m$) in the case of regular topology, and (ii) In the $k$-random topology, while each cheetah $i$ is still connected to $k$ gazelles, these $k$ gazelles are selected randomly from the total $m$ gazelles. It is important to note that each cheetah \textit{is only assumed to be aware of the $k$ gazelles that she is connected to, and not of the gazelles any other cheetah may access}.

We shall assume that the gazelle weights are identical. Moreover, attach some cost $v$ to hunting a particular gazelle. All other assumptions of the original model are preserved. 

Let each cheetah not hunt at all with probability $1-p$, and let the probability that each cheetah hunts a given gazelle accessible to it be $\frac{p}{k}$. We seek to find the conditions where this is an equilibrium. 

\begin{lemma}
    In a regular small-world topology as described, the equilibrium probability $p^\ast$ of hunting is: 
$$ p^\ast = \begin{cases}
1 & \text{if }v < \left(1-\frac{1}{k}+\frac{r}{k}\right)^{k-1}\\
\frac{k(1-v^{1/(k-1)})}{1-r} & \text{otherwise}
\end{cases}$$

Moreover, the chance that a given gazelle survives is: 
    $$ \left(1 - \frac{p^\ast}{k}\right)^k   $$
\end{lemma}

\begin{proof}
Assume that each cheetah hunts some gazelle with probability $p$ when in equilibrium. Evidently, cheetah 1's expected winnings $W$ by sending all resources to $1$ is:
$$ W = (1 - p + \frac{k-1}{k} p )^{k-1} + \binom{k}{1} (1 - p + \frac{k-1}{k} p )^{k-2} \frac{1}{k} r^2 \dots$$
By the binomial theorem,
$$ W(p) = (1 - p + \frac{k-1}{k} p + \frac{rp}{k})^{k-1}  = (1 - \frac{1-r}{k} p)^{k-1}$$

 Now, if $w \ge W(1)$, any given cheetah mixes on the strategies of not hunting and hunting, each with positive, probability. Hence, in equilibrium:
$$0 = W(p) - v $$
$$ \implies v = (1 - \frac{1-r}{k} p)^{k-1}$$ 
$$ \implies p = k(\frac{1 - v^{1/{k-1}}}{1-r}) $$
Hence the probability that no cheetahs hunt a given gazelle (and hence that the gazelle survives) is:
$$(1 - \frac{p}{k})^k$$

 On the other hand, if $v <  W(1)$, then it is optimal to leave with probability $p = 1$. Therefore, the probability that no cheetahs hunt a given gazelle is:
$(1- 1/k)^k$, which is the survival probability.
\end{proof}

We also claim:
\begin{lemma}
 In a random topology, regardless of $k$, the equilibrium probability $p^\ast$ of hunting is: 
$$ p^\ast = \begin{cases}
1 & \text{if }v < \left(1-\frac{1}{m}+\frac{r}{m}\right)^{n-1}\\
\frac{m(1-v^{1/(n-1)})}{1-r} & \text{otherwise}
\end{cases}$$
Moreover, when $m = n,$ as $ n\rightarrow \infty$, $p^\ast \rightarrow \frac{-\ln v}{1-r}$.

The chance that a given gazelle survives is: $$ \left(1 - \frac{p^\ast}{m}\right)^n =  
\begin{cases}
\left(1 - \frac{1}{m}\right)^n & \text{if }v < \left(1-\frac{1}{m}+\frac{r}{m}\right)^{n-1}\\
(1 - \frac{1-v^{1/(n-1)}}{1-r})^n & \text{otherwise}
\end{cases}
$$
\end{lemma}

 \begin{proof}
Assume that any given cheetah hunts with probability $p$, in equilibrium. 
Note that the probability that any given cheetah hunts gazelle $1$ is: $$ \left(\frac{p}{k} \right) * (\text {Chance that at least one of the $k$ gazelles picked is 1}) = \frac{p}{k} \cdot \frac{\dbinom{m-1}{k-1}}{\dbinom{m}{k}} = \frac{p}{k} \cdot \frac{k}{m} = \frac{p}{m}$$
The chance that $t$ cheetahs send resources to gazelle $1$ is clearly $$\frac{p}{m}^t \left(1 - \frac{p}{m} \right)^{n-1-t}$$

 Hence, if cheetah 1 sends all resources to $1$, the expected payoff $W$ satisfies: 
$$ W(p) = \sum_{t=0}^n \frac{p}{m}^t \left(1 - \frac{p}{m} \right)^{n-t} (r^t) =  \sum_{t=0}^n \frac{pr}{m}^t \left(1 - \frac{p}{m} \right)^{n-t} = \left(1 - \frac{p}{m} +  \frac{pr}{m} \right)^ {n-1} $$

Hence,
$$ W = \left(1 - \frac {1-r}{m} p\right) ^ {n-1} $$
 This demonstrates that the expected payoff is not dependent on $k$.

Now, there are 2 cases for $v$:

1) $ v < W(1)$. In this case, the payoff will be higher if the cheetah hunts with probability $1$, so $p = 1$. 

2) On the other hand, if $v \ge W(1)$, for equilibrium to be achieved, $v= W(p)$, so:
$$ p = \frac{m}{1-r}(1 - v^{1/{(n-1)}}) $$ As $n \to \infty$ and $m \to n$, this approaches:
$$ \frac{1}{1-r} \lim_{x \to 0} \frac{1 - v^x}{x} $$
By L'Hospital's Rule, this is:
$$ p = \frac{1}{1-r} (\lim_{x \to 0} \frac{1 - v^x}{x}) = \frac{-1}{1-r}(\lim_{x \to 0} \frac{ v^x  \ln v}{1}) = \frac{-\ln v}{1-r}$$

 So, the probability that gazelle $j$ survives (or the probability that no cheetah hunts gazelle $j$) is:
$$(1 - \frac{p}{m})^n = (1 - \frac{1-v^{1/(n-1)}}{1-r})^n$$ \end{proof}

From now on, we assume $m \to n$ so as to compare the random and regular topology survival rates.
Let the survival rate be $S_{rnd}$ in a random topology and $S_{reg}$ in a regular topology. Similarly, let $A_{rnd}$ and $A_{reg}$ represent the expected payoff gained by hunting a particular gazelle in the random and regular topologies, respectively.

Consider the following cases:

\noindent Case 1) $v < A_{rnd}, A_{reg}$

In this case, $S_{reg} = (1 - \frac{1}{k})^k$, while $S_{rnd} = (1 - \frac{1}{n})^n$. Hence, $S_{reg} < S_{rnd}$

\noindent Case 2)$v \in (A_{rnd}, A_{reg})$

Consider the graph of survival probability against $v$, with all other variables fixed. In the random case, the graph is constant until some critical value $v_1$, while in the regular case, the graph is constant until some value $v_2$. If, at some point, $S_{rnd} < S_{reg}$, even though $S_{reg} < S_{rnd}$ holds at $v_1, v_2$, $S_{reg}$ must cross a constant value twice. This is not possible, as it is increasing. Hence, $S_{reg} < S_{rnd}$

\noindent Case 3) $v \geq A_{reg}, A_{rnd}$

Clearly, in the case of a regular topology, the probability of a gazelle surviving is: $f(k)$, while in the case of a random topology, the probability of a gazelle surviving is $f(n)$, where
$f(x) = (1 - \frac{1- v^{1/(x-1)}}{1-r})^x$.

 Since $f(x)$ is increasing, we have that: \textit{the survival rate is greater in a random topology than in a ``small-world'' regular topology}

The graph below shows $f(x) = (1 - \frac{1- v^{1/(x-1)}}{1-r})^x$ for given values of $v$ and $r$. As can be seen, in the last case, survival probability increases as the number of accessible gazelles decreases.

Moreover, by a simple casework argument, this implies that the average payoff of cheetahs is higher in regular, ``small-world'' networks. This is proved in the supplementary appendix, as Lemma 5.

\begin{center}
    \includegraphics[scale=0.3]{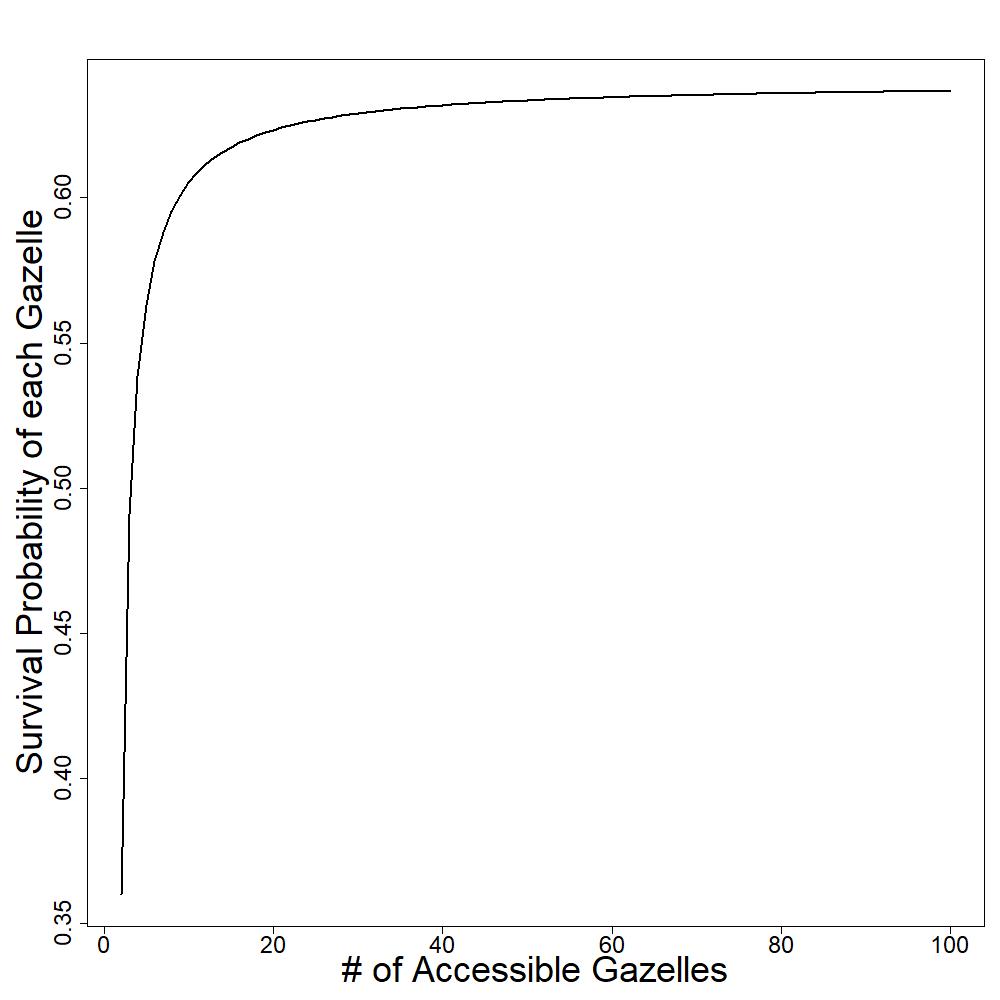}
\end{center}




\section {\label{discussion-and-conclusions} Discussion \& Conclusions}
This study of the network Colonel Blotto game has yielded three main results: First, the existence and uniqueness of interior point equilibrium solutions for a certain class of regular network topologies is determined by the residue (mod 4) of the number of vertices of the network. This result is technically intriguing, since it seems to make a connection between what would, at first sight, appear to be a purely number theoretical property of a network and existence/uniqueness of equilibrium of a resource allocation contest.

Secondly, we proved that if the accessibility network is randomly generated, the survival rate is higher compared to the case where the accessibility network is regular. In very rough terms, if the game is made more "local" through restrictions on movement, where each field can only be accessed by a small set of players, then, many more fields will be ``invaded'' by the players on expectation. Thus, while each individual cheetah may choose to have a greater "range" (i.e., be able to hunt more gazelles), cheetahs are better off when their ranges are collectively constrained.

Finally, we discovered that heterogeneity in field weights leads to lower average payoffs and higher (average) survival rates (i.e., more fields that are "untouched"). While the second part of this is somewhat more intuitive (of course fields with extremely low valuation are avoided), the first part is quite unexpected.



The current study focused on an ``all-or-nothing'' variant (where each player had a single indivisible unit resource). In the future, it would be useful to extend this to cases where strengths of each player vary. This will allow the model to be applied to various real-life strategy games where resources must be allocated in a competition against an unknown competitor. 

For example, in online phishing scams (which have been studied using the general network-free Colonel Blotto game in the past by \cite{chia2011colonel}), scammers are competing to send the maximum number of phishing emails to various email addresses, most of which are picked at random. The study of such scams would be facilitated by the network Blotto model introduced in this paper, since scammers have access to only a certain subset of email records, and as a result can only send emails to a certain fraction of accounts. Additionally, future research should change the tie-breaking function to allow analysis of variants where ties between players are broken randomly (i.e., one of the players that sends the most resources is designated, uniformly at random, a winner). This would allow more accurate analysis of real-life discrete scenarios, such as various situations in evolutionary biology, where species compete over a subset of the total number of niches, and no two species can occupy the same niche in the long term.

\bibliography{blotto}

\section*{Supplementary Appendix} 
We prove the following claim to confirm the main prediction of Hypothesis 1:

\begin{lemma}
    When $m \to n$, the equilibrium average payoff in the random case is independent of $k$, and is not more than the  equilibrium average payoff in the regular case, for any regular accessibility network.

\end{lemma}
\begin{proof}
Let $p$ be the equilibrium probability of hunting, for every cheetah. Let $f(k) = (1 - 1/k + r/k)^{k-1}$.

Firstly, note that the number of cheetahs $N$ that a given cheetah ties with if she visits a given gazelle satisfies $N \sim \text{Binom}{(n, p/k)}$

Hence, the expected net utility by hunting a given gazelle is:
\begin{equation}
    -v+ \mathbb{E}(r^N) = -v + \mathbb{E} (e^{N \ln r}) = -v + (1 - \frac{p}{k} + \frac{rp}{k})^{k-1} 
\end{equation}

The last line derives from the fact that $\mathbb{E} (e^{Nt})$ is the moment generating function for a binomially-distributed random variable $N$

Hence, in a regular topology:
$$ p = \begin{cases}
1; \hspace{8pt} v < f(k) \\
p^* = \dfrac{k (1- v^{1/(k-1)}) }{1-r}; \hspace{8pt} v \ge f(k) \\
\end{cases} $$

On the other hand, in a random toplogy, when $m = n$,  $N \sim \text{Binom} (n-1, p/n)$, so the expected net utility by hunting a given gazelle is:
$$-v + \mathbb{E} (r^N) = -v + \mathbb{E}(e^{N \ln r}) = -v + (1 - \frac{p}{n} + \frac{rp}{n})^{n-1} $$

Therefore, in a random topology, from Lemma 4:
$$ p = \begin{cases}
1; \hspace{8pt} v < f(n) \\
p^* = \dfrac{n (1- v^{1/(n-1)}) }{1-r}; \hspace{8pt} v \ge f(n) \\
\end{cases} $$

Hence, we have 3 cases for the relative position of $v$:

1) $v < f(n)$:
In this case, since $f(k) < f(n)$, $p_{reg} = 1$, and $p_{rnd} = 1$. 

So, the expected net winnings per cheetah in the random topology is:
$$W_{rnd} = \mathbb{E}(r^N) - v = ( 1- \frac{1-r}{n})^{n-1}- v $$ 

The expected net winnings in a regular topology is:
$$ W_{reg} = \mathbb{E}(r^N) - v = ( 1- \frac{1-r}{k})^{k-1}- v$$

As $f(x) = (1- \frac{1-r}{x})^{x-1}$ is decreasing, $W_{reg} > W_{rnd}$

2) $v \geq f(k)$
In this case, a cheetah in the random topology has expected net winnings:
$$W_{rnd} = p_{rnd}(\mathbb{E}(r^N)) - p_{rnd}v = p_{rnd} (1 - \frac{p_{rnd}}{n} + \frac{rp_{rnd}}{n})^{n-1} - p_{rnd}v = 0$$

In the regular topology, the expected net payoff for a given cheetah is:
$$W_{reg} = p_{reg}(\mathbb{E}(r^N)) - p_{reg}v = p_{reg} (1 - \frac{p_{reg}}{k} + \frac{rp_{reg}}{k})^{k-1} - p_{reg}v = 0$$

So, the expected payoff per cheetah is equal in the random and regular topologies.

3) $f(n) < v < f(k)$

From the previous two cases, $W_{reg} = \mathbb{E}(r^N) - v = ( 1- \frac{1-r}{k})^{k-1}- v \geq 0$, while $W_{rnd} = 0$.

Putting all of these cases together, $W_{reg} \geq W_{rnd}$. The lemma follows.
\end{proof}

\end{document}